\documentclass[aps,prd,reprint,amsmath,amssymb,nofootinbib]{revtex4-2}

\usepackage{graphicx}
\usepackage{bm}
\usepackage{mathtools}
\usepackage{mathrsfs}
\usepackage{booktabs}
\usepackage{xcolor}
\usepackage[colorlinks=true,allcolors=blue!55!black]{hyperref}

\newtheorem{theorem}{Theorem}
\newtheorem{proposition}[theorem]{Proposition}

\newtheorem{corollary}[theorem]{Corollary}
\newtheorem{definition}[theorem]{Definition}

\newenvironment{proof}[1][Proof]
 {\par\noindent\textit{#1.}\ }
 {\hfill$\square$\par}

\newcommand{\dd}{\mathrm{d}}

\newcommand{\scri}{\mathscr{I}}
\newcommand{\order}{\mathcal{O}}
\newcommand{\cA}{\mathcal{A}}
\newcommand{\cF}{\mathcal{F}}
\newcommand{\cH}{\mathcal{H}}
\newcommand{\cT}{\mathscr{T}}

\begin{document}

\title{Optimal Extension Regularity at the McVittie Event Horizon}

\author{Yi-kun Li}
\email{liyikun@xao.ac.cn}
\affiliation{State Key Laboratory of Radio Astronomy and Technology,
Xinjiang Astronomical Observatory, CAS, 150 Science 1-Street,
Urumqi 830011, China}
\affiliation{School of Astronomy and Space Science,
University of Chinese Academy of Sciences, No.19A Yuquan Road,
Beijing 100049, China}

\date{August 13, 2026}

\begin{abstract}
We determine the optimal local extension regularity of the future black-hole
event horizon in the exact spatially flat McVittie solutions sourced by a
positive cosmological constant and a barotropic fluid with constant
equation-of-state parameter $w>-1$.  Let $H_\infty$ be the asymptotic Hubble
constant, $\kappa$ the surface gravity of the limiting black-hole root, and
$p=3(1+w)H_\infty/\kappa$.  Ingoing radial null geodesics reach the horizon in
finite affine length.  A parallelly propagated angular curvature component is
asymptotic to $C s^{p-2}$, with $C\ne0$ and $s$ the remaining affine distance,
which excludes every anchored $C^2$ extension for $0<p<2$.  For $p\ge2$ we
construct a parameter-uniform Gaussian-null compactification and an explicit
two-sided Lorentzian collar.  If $p=N+\vartheta$ is nonintegral, with $N\ge2$
and $0<\vartheta<1$, the optimal regularity is the standard big H\"older class
$C^{N,\vartheta}$: extensions of this class exist, whereas no
$C^{N,\vartheta'}$ extension exists for $\vartheta'>\vartheta$.  Every integer
$p\ge2$ instead belongs to an analytic island and admits a real-analytic local
extension.  At the critical value $p=2$ the boundary Einstein endomorphism
has a nonzero rank-one nilpotent part.  The ratio of cosmological decay to
horizon redshift therefore determines a sharp, arithmetic hierarchy of
geometric regularity.
\end{abstract}

\maketitle

\section{Introduction}
\label{sec:introduction}

The differentiability of a black-hole boundary is a geometric question about
how the physical spacetime approaches the boundary and which tidal fields an
affine observer measures there.  It becomes especially delicate when the
boundary lies at infinite cosmological time but finite affine distance.  A
metric can converge pointwise to a stationary black-hole geometry while its
transverse derivatives retain the decay rate of the surrounding matter.  The
competition between this decay and the horizon redshift then determines the
regularity of any completion.

McVittie's solution provides an exact setting in which this competition can
be resolved.  The metric embeds a central Schwarzschild mass in a homogeneous
FLRW background without radial accretion \cite{McVittie1933,Nolan1998}.
For cosmologies tending to de Sitter space, the limiting inner surface has
been identified and interpreted through complementary analyses of radial null
geodesics, causal structure, and explicit conformal extensions
\cite{KaloperKlebanMartin2010,LakeAbdelqader2011,
DaSilvaFontaniniGuariento2013,Nolan2014}.  These studies also show that the
rate at which the Hubble function approaches its limit affects the causal
diagram.  In particular, the radial characteristic equation has a resonance
when the cosmological decay rate equals the redshift rate of the limiting
root.

The regularity of the finite-affine endpoint requires further information.
Scalar curvature invariants remain finite in the positive-$\Lambda$ regime,
yet scalar contractions do not control curvature in a parallelly propagated
frame \cite{EllisSchmidt1979}.  The recent analysis of timelike observers by
Nolan found finite p.p. curvature and regular Jacobi fields for the trajectories
considered there \cite{Nolan2026}.  Radial null observers occupy a distinct
boost sector: their affine tangent grows exponentially relative to the
cosmological slicing, and the boost acts twice on the transverse tidal tensor.
This produces a differentiability threshold that is separated from the causal
resonance by one power of affine distance.

We study the exact spatially flat background formed by a positive cosmological
constant and a perfect fluid with constant $w>-1$.  If $H_\infty$ is the
asymptotic Hubble constant and $\kappa>0$ the surface gravity of the limiting
black-hole root, the relevant dimensionless exponent is
\begin{equation}
 p=\frac{3(1+w)H_\infty}{\kappa}.
 \label{eq:intro-p}
\end{equation}
The numerator is the leading exponential decay rate of the barotropic density;
the denominator converts cosmological time into remaining affine distance.
The same exponent appears in the late-time expansion of Kaloper, Kleban, and
Martin \cite{KaloperKlebanMartin2010}.  Here it is promoted to a complete local
extension classification by combining an invariant obstruction with an
explicit construction.

The construction begins with a compact family of ingoing radial null
geodesics.  Their asymptotic amplitude gives a longitudinal coordinate, while
an integral of the parameter Jacobi field gives a transverse affine
coordinate.  In the resulting Gaussian-null chart, the physical metric has a
parameter-uniform polyhomogeneous expansion with index set
\begin{equation}
 E_p=\{j+np:(j,n)\in\mathbb N_0^2,\ (j,n)\ne(0,0)\}.
 \label{eq:intro-index}
\end{equation}
The unique indicial root is occupied by the freely chosen geodesic amplitude,
and every forced exponent is larger than one.  Consequently no logarithm is
generated for $p\ge2$.  Finite reflection operators then extend the
Gaussian-null coefficients through the endpoint with their attained H\"older
regularity.  When $p$ is an integer, the characteristic equation becomes an
analytic parameter-dependent ODE and yields a convergent two-sided analytic
germ.

Sharpness comes from the angular tidal scalar along the anchored affine
generator.  Its first nonstationary term is
\begin{align}
 \cT&=R_{kAkA}
 =P(\rho)+K(\alpha,p)\rho^{p-2}
 +\order(\rho^{p-1}),\nonumber\\
 K(\alpha,p)&\ne0,
 \label{eq:intro-tidal}
\end{align}
where $P$ is a finite integer-power Taylor polynomial.  Curvature continuity
gives the $C^2$ obstruction for $p<2$.  For nonintegral
$p=N+\vartheta\ge2$, differentiating $N-2$ times leaves a nonzero
$\rho^\vartheta$ term and excludes every higher H\"older exponent.  Anchoring
the extension to the geodesic makes this conclusion independent of the
endpoint chart and follows the geometric logic used in modern spacetime
inextendibility results \cite{Sbierski2024Uniqueness,Sbierski2025Curvature}.

The critical endpoint also records matter information.  At $p=2$ the mixed
Einstein tensor equals its de Sitter value plus a nonzero rank-one nilpotent
endomorphism, giving a type-II null limit.  Two exact McVittie open regions can
nevertheless be joined with a $C^2$ metric, continuous stress tensor, and no
surface layer.  The algebraic boundary value and the open-side field equations
therefore describe complementary aspects of the same completion.

Section~\ref{sec:spacetime} introduces the exact background and the global
event-horizon identification.  Section~\ref{sec:null} obtains the affine null
frame and the curvature obstruction.  Section~\ref{sec:extension} constructs
the Gaussian-null collar, and Sec.~\ref{sec:regularity} proves the optimal
H\"older hierarchy and the analytic islands.  Matter limits and physical
consequences are discussed in Sec.~\ref{sec:consequences}; technical estimates
and the exact-copy construction are collected in the appendices.

\section{McVittie spacetime and its event horizon}
\label{sec:spacetime}

\subsection{Exact cosmological background and matter}

We use geometrized units $G=c=1$ and the signature $(-,+,+,+)$.  The spatially flat McVittie line element in isotropic coordinates $(t,r,\theta,\phi)$ is
\begin{align}
 \dd s^2={}&-\left(\frac{1-\mu}{1+\mu}\right)^2\dd t^2
 +a(t)^2(1+\mu)^4(\dd r^2+r^2\dd\Omega^2),\nonumber\\
 \mu={}&\frac{m}{2a(t)r},
 \label{eq:mcvittie-isotropic}
\end{align}
where $m>0$ is constant and $\dd\Omega^2$ is the metric on the unit sphere.  The constancy of $m$ is the no-accretion condition.  Passing to areal radius
\begin{equation}
 R=a r(1+\mu)^2
 \label{eq:areal-radius}
\end{equation}
and introducing
\begin{equation}
 x=\frac{R}{m},\qquad \tau=\frac{t}{m},\qquad
 h=mH,\qquad S(x)=1-\frac{2}{x},
 \label{eq:dimensionless-defs}
\end{equation}
gives the dimensionless form
\begin{align}
 \frac{\dd s^2}{m^2}={}&-f\,\dd\tau^2
 -\frac{2hx}{\sqrt S}\,\dd\tau\dd x
 +\frac{\dd x^2}{S}+x^2\dd\Omega^2,\nonumber\\
 f={}&S-h^2x^2.
 \label{eq:mcvittie-areal}
\end{align}
The determinant of the two-dimensional orbit metric is $-1$.  This elementary identity simplifies both the radial geodesic equations and the Gaussian-null construction.
The exterior portion of the original McVittie chart has $x>2$.  Constant-$x$ worldlines are generally accelerated, while the fluid four-velocity is orthogonal to the homogeneous time slices in isotropic coordinates.  In the fluid rest frame the radial energy current vanishes, and the off-diagonal term in Eq.~\eqref{eq:mcvittie-areal} arises from the use of areal radius.  Keeping this term near the limiting horizon retains the decay information carried by the dynamical geometry.

The exact cosmological source is a positive cosmological constant together with a perfect fluid satisfying $P_w=w\rho_w$ with constant $w>-1$.  A convenient normalization is
\begin{align}
 \frac{a(\tau)}{a_*}&=\sinh^{\frac{2}{3(1+w)}}z,
 &z&=\frac{3}{2}(1+w)h_\infty\tau, \label{eq:scale-factor}\\
 h(\tau)&=h_\infty\coth z,
 &\dot h&=-\frac{3}{2}(1+w)(h^2-h_\infty^2),
 \label{eq:hubble-exact}
\end{align}
where a dot denotes $\dd/\dd\tau$ and $h_\infty=mH_\infty>0$.  The big bang lies at $\tau=0$, while $h\to h_\infty$ as $\tau\to\infty$.  Einstein's equations for Eq.~\eqref{eq:mcvittie-areal} yield
\begin{equation}
 8\pi\rho=3H^2,\qquad
 8\pi P=-3H^2-\frac{2\dot H}{\sqrt{1-2m/R}}.
 \label{eq:mcvittie-fluid}
\end{equation}
The mixed Einstein tensor has one timelike eigenvalue and a threefold spatial eigenvalue, and the energy flux measured in the fluid rest frame vanishes.  Equation~\eqref{eq:mcvittie-fluid} also shows that the fluid approaches vacuum energy as $\tau\to\infty$ at every fixed $x>2$.
The density is homogeneous, while the pressure acquires the factor $S^{-1/2}$ required by hydrostatic balance around the central mass.  Decomposing $\rho=\rho_\Lambda+\rho_w$ gives $\rho_\Lambda=3H_\infty^2/(8\pi)$ and $P_\Lambda=-\rho_\Lambda$.  The remaining component obeys $P_w=w\rho_w$ in the cosmological region; the inhomogeneous local pressure in Eq.~\eqref{eq:mcvittie-fluid} follows from the same background $\dot H$.  As $H-H_\infty$ decays, both the density contrast and this pressure correction vanish at the limiting horizon.  Their decay rate controls differentiability.

The exact late-time variable
\begin{equation}
 y=e^{-3(1+w)h_\infty\tau}
 \label{eq:ydef}
\end{equation}
puts the Hubble tail in the useful form
\begin{equation}
 h-h_\infty=\frac{2h_\infty y}{1-y}
 =2h_\infty\sum_{n=1}^{\infty}y^n.
 \label{eq:coth-tail}
\end{equation}
Thus the leading dimensionless decay rate is
\begin{equation}
 \bar\lambda=3(1+w)h_\infty.
 \label{eq:lambda-bar}
\end{equation}
The full geometric series in Eq.~\eqref{eq:coth-tail} controls the borderline normal form in Sec.~\ref{sec:extension}.

\subsection{Limiting roots and surface gravity}

The scalar $\nabla_aR\nabla^aR$ equals $f$, equivalently $1-2M_{\rm MS}/R$ in terms of the Misner--Sharp mass \cite{MisnerSharp1964}, so the marginal spheres satisfy $f=0$.  At late times,
\begin{equation}
 f_\infty(x)=1-\frac{2}{x}-h_\infty^2x^2.
 \label{eq:f-infinity}
\end{equation}
For
\begin{equation}
 0<h_\infty^2<\frac{1}{27},
 \label{eq:subnariai}
\end{equation}
Eq.~\eqref{eq:f-infinity} has two positive simple roots.  We denote the smaller one by $\alpha$.  It is the unique root in $(2,3)$ and obeys
\begin{equation}
 h_\infty^2=\frac{\alpha-2}{\alpha^3}.
 \label{eq:alpha-relation}
\end{equation}
The dimensionless signed surface gravity at this root is
\begin{equation}
 \bar\kappa=\frac{1}{2}f_\infty'(\alpha)
 =\frac{3-\alpha}{\alpha^2}>0.
 \label{eq:kappa-bar}
\end{equation}
The larger root is cosmological and has the opposite sign of $f_\infty'$.  At $h_\infty^2=1/27$ the two roots coalesce at $x=3$; the simple-root asymptotics used below are then replaced by the degenerate Nariai scaling \cite{StuchlikHledik1999}.

The time-dependent black-hole marginal tube is the smaller solution $x_-(\tau)$ of $f(\tau,x)=0$.  Since $h$ decreases monotonically, $x_-(\tau)$ approaches $\alpha$ from above after the two positive roots form.  Its motion supplies a useful geometric picture.  The equation $f=0$ locates the marginal tube, while causal reachability identifies the event horizon \cite{FaraoniMorenoNandra2012}.
At the creation time of the pair, the two roots are degenerate.  Thereafter the regular region $f>0$ lies between the black-hole and cosmological marginal tubes.  Outgoing radial light rays increase in areal radius throughout this region, while the ingoing family can either move inward or be carried outward by expansion according to the sign of $D=\sqrt S-hx$.  The limiting root $\alpha$ is reached only at $\tau=\infty$ in the areal chart.  Its finite affine accessibility is established from the boost equation in Sec.~\ref{sec:null}.

\subsection{Identification of the future event horizon}

The global properties needed here follow from the radial-null results proved for expanding big-bang McVittie spacetimes \cite{Nolan2014,Nolan2026}.  In particular, Proposition 1.1 of Ref.~\cite{Nolan2026} states the future completeness of the outgoing family and the finite-affine endpoint of the ingoing family.  The exact background in Eqs.~\eqref{eq:scale-factor} and \eqref{eq:hubble-exact} satisfies its hypotheses: $H>0$, $\dot H<0$, the regular region is nonempty at sufficiently late time, $H$ has the positive limit $H_\infty$, and the limiting black-hole root is simple under Eq.~\eqref{eq:subnariai}.  Outgoing radial null geodesics in the regular exterior escape to the expanding cosmological end, whereas the ingoing family considered below reaches $(\tau,x)=(\infty,\alpha)$ in finite affine parameter.

Relative to that future cosmological end, the limiting null boundary is therefore
\begin{equation}
 \mathcal H^+=\partial J^-(\scri^+_{\mathcal E}),
 \label{eq:event-horizon}
\end{equation}
the future black-hole event horizon \cite{HawkingEllis1973,Wald1984}.  This identification supplies the global anchor for the local differentiability problem.  The analysis of Secs.~\ref{sec:null} and \ref{sec:extension} follows the incomplete ingoing generators to their endpoints on $\mathcal H^+$ and determines their radial-null tidal behavior there.  The invariant location of this problem is the finite affine endpoint reached as $\tau\to\infty$.
The hypotheses also explain the role of the exact barotropic family.  The big-bang origin fixes the expanding branch, monotonicity of $H$ prevents late-time oscillations of the marginal tubes, and the positive limit creates a stationary reference end.  Simplicity of the black-hole root supplies an exponential redshift and a nonzero $\bar\kappa$.  These properties place the whole interval $w>-1$ under the same causal theorem while allowing the late-time matter exponent to vary continuously.  They provide a common global exterior for comparing the causal and curvature thresholds.

\section{Affine null geometry and tidal regularity}
\label{sec:null}

\subsection{Radial null system and parallelly propagated frame}

It is useful to retain the affine scale when taking the late-time limit.  Let $\ell=\lambda/m$ be a dimensionless affine parameter and define
\begin{equation}
 D(\tau,x)=\sqrt{S(x)}-h(\tau)x.
 \label{eq:Ddef}
\end{equation}
The future-directed ingoing branch of the radial null equation is
\begin{equation}
 \frac{\dd x}{\dd\tau}=-\sqrt S\,D.
 \label{eq:ingoing-null-slope}
\end{equation}
Writing $\mathcal Q=\dd\tau/\dd\ell$, its affinely parametrized tangent is
\begin{equation}
 k=\mathcal Q\left(\partial_\tau-\sqrt S D\,\partial_x\right),
 \label{eq:k-vector}
\end{equation}
and the remaining geodesic equation reduces to
\begin{equation}
 \frac{\dd\mathcal Q}{\dd\ell}
 =\mathcal Q^2\left[\frac{2}{x^2}
 -\frac{h(x-1)}{x\sqrt S}\right].
 \label{eq:Q-equation}
\end{equation}
Equations~\eqref{eq:ingoing-null-slope} and \eqref{eq:Q-equation} imply the null constraint and both coordinate geodesic equations.  A constant rescaling of $\ell$ rescales $\mathcal Q$ and the amplitudes below while preserving their exponents and nonvanishing character.
The function $D$ vanishes on the time-dependent black-hole marginal tube.  Near the limiting root it is linear in $x-\alpha$ and in $h-h_\infty$, so Eq.~\eqref{eq:ingoing-null-slope} contains both the stationary redshift and the cosmological forcing.  The affine equation supplies the corresponding boost information: $\mathcal Q$ grows as the horizon is approached, encoding the relation between an affine photon frame and the cosmological time slicing.  Fixing $-u\cdot k=1$ at any regular exterior event determines this boost up to the choice of that event and leaves the threshold invariant.

A complementary null vector with $k\cdot n=-1$ is
\begin{equation}
 n=\frac{1}{2S\mathcal Q}\partial_\tau
 +\frac{\sqrt S(hx+\sqrt S)}{2S\mathcal Q}\partial_x.
 \label{eq:n-vector}
\end{equation}
Together with
\begin{equation}
 e_{\hat\theta}=\frac{1}{x}\partial_\theta,
 \qquad
 e_{\hat\phi}=\frac{1}{x\sin\theta}\partial_\phi,
 \label{eq:angular-frame}
\end{equation}
these vectors form a null orthonormal frame.  Direct use of the connection of Eq.~\eqref{eq:mcvittie-areal} gives
\begin{equation}
 \nabla_k k=\nabla_k n=\nabla_k e_{\hat\theta}
 =\nabla_k e_{\hat\phi}=0.
 \label{eq:parallel-frame}
\end{equation}
The angular vectors are parallelly transported because the change of the unit-sphere basis is exactly canceled by the radial warp connection.  Appendix~\ref{app:null-frame} gives the component calculation and an independent warped-product derivation.
Parallel propagation is the crucial normalization for a tidal statement.  Pointwise orthonormality leaves a position-dependent boost freedom, which Eq.~\eqref{eq:parallel-frame} fixes along the generator.  The resulting frame is regular for every finite $s>0$, and its normalization is transported unchanged to the incomplete endpoint.

\subsection{Complete radial-null tidal tensor}

We quote dimensionless frame components $\mathcal R_{IJKL}=m^2R_{abcd}e_I^ae_J^be_K^ce_L^d$.  Spherical symmetry leaves five independent entries along a radial null geodesic:
\begin{align}
 \mathcal R_{knkn}&=-\frac{2}{x^3}-h^2-\frac{\dot h}{\sqrt S},
 \label{eq:Rknkn}\\
 \mathcal R_{k\hat A k\hat B}
 &=-\mathcal Q^2\sqrt S\,\dot h\,\delta_{AB}
 =-\frac{1}{x}\frac{\dd^2x}{\dd\ell^2}\delta_{AB},
 \label{eq:RkAkB}\\
 \mathcal R_{n\hat A n\hat B}
 &=-\frac{\dot h}{4\mathcal Q^2S^{3/2}}\delta_{AB},
 \label{eq:RnAnB}\\
 \mathcal R_{k\hat A n\hat B}
 &=\left(\frac{1}{x^3}-h^2-\frac{\dot h}{2\sqrt S}\right)\delta_{AB},
 \label{eq:RkAnB}\\
 \mathcal R_{\hat\theta\hat\phi\hat\theta\hat\phi}
 &=h^2+\frac{2}{x^3}.
 \label{eq:Rangular}
\end{align}
Here $A,B\in\{\hat\theta,\hat\phi\}$.  The equality in Eq.~\eqref{eq:RkAkB} is the spherical warped-product identity
\begin{equation}
 R_{kAkB}=-\frac{1}{R}\frac{\dd^2R}{\dd\lambda^2}\delta_{AB}.
 \label{eq:warped-identity}
\end{equation}
It provides a derivation independent of the four-dimensional Riemann contraction.  Polynomial scalar invariants formed from the Riemann tensor remain finite as $(\tau,x)\to(\infty,\alpha)$ because $h\to h_\infty$, $\dot h\to0$, and $S(\alpha)>0$.  The factor $\mathcal Q^2$ in Eq.~\eqref{eq:RkAkB} records the unbounded boost of the affine null frame relative to the cosmological fluid.
The five entries give the complete radial-null tidal tensor.  The boost-enhanced angular component is sensitive to the time derivative of the Hubble function.  Its conjugate component carries $\mathcal Q^{-2}$ and decays, while the mixed and purely orbital components tend to their Schwarzschild--de Sitter values.  The independence of these frame components leaves the divergence in Eq.~\eqref{eq:RkAkB} intact.  The warped-product equality in Eq.~\eqref{eq:warped-identity} also identifies its observable effect: it is the relative angular acceleration of neighboring radial photons.

\subsection{Sharp asymptotics and the obstruction for \texorpdfstring{$p<2$}{p<2}}

Let $\ell_h$ be the finite endpoint of an ingoing generator and set
\begin{equation}
 s=\ell_h-\ell>0.
 \label{eq:s-def}
\end{equation}
Linearization of Eqs.~\eqref{eq:ingoing-null-slope} and \eqref{eq:Q-equation} at the simple limiting root gives
\begin{equation}
 s=s_0e^{-\bar\kappa\tau}[1+o(1)],
 \qquad
 \mathcal Q=Q_0s^{-1}[1+o(1)],
 \label{eq:s-Q-asymptotic}
\end{equation}
with $s_0,Q_0>0$ after fixing the affine orientation.  The exact background equations give
\begin{equation}
 \dot h=-2h_\infty\bar\lambda e^{-\bar\lambda\tau}
 [1+o(1)].
 \label{eq:hdot-tail}
\end{equation}
Combining the two expressions defines
\begin{equation}
 p=\frac{\bar\lambda}{\bar\kappa}
 =\frac{3(1+w)h_\infty}{\bar\kappa}
 \label{eq:p-def}
\end{equation}
and yields the leading parallelly propagated curvature
\begin{align}
 \mathcal R_{k\hat A k\hat B}
 &=C\,s^{p-2}\delta_{AB}+o(s^{p-2}),\nonumber\\
 C&=2Q_0^2\sqrt{S(\alpha)}h_\infty\bar\lambda s_0^{-p}>0.
 \label{eq:pp-leading}
\end{align}
The normalization-dependent factors in $C$ are finite and positive.  Its nonzero sign follows from $w>-1$, $h_\infty>0$, and the simple-root assumptions.
Every harmonic in the exact coth tail produces a higher power $s^{np-2}$.  The leading $n=1$ term therefore decides the curvature classification throughout the family.  At $p=2$ it tends to a finite nonzero value; for $p>2$ it vanishes, leaving the stationary horizon curvature.  At $p=1$ the power law in Eq.~\eqref{eq:pp-leading} remains regular in its exponent, while the forced radial displacement becomes resonant in horizon-adapted coordinates.

For $0<p<2$, Eq.~\eqref{eq:pp-leading} diverges at a finite affine endpoint.  If a $C^2$ Lorentz metric extended the spacetime through that endpoint, its Riemann tensor would be continuous in a regular frame and its contraction with the parallelly transported vectors would remain finite.  The divergence therefore rules out a $C^2$ extension containing this generator.  This is an anchored obstruction: the curve, affine endpoint, and transported frame are all fixed by the original spacetime.
More explicitly, a hypothetical $C^2$ extension would continue the geodesic and its parallel-transport equation with continuous Christoffel symbols.  The frame would possess a finite limit, and continuity of the extended Riemann tensor would bound each contraction $R_{IJKL}$.  Equation~\eqref{eq:pp-leading} contradicts that bound.  This argument is insensitive to the areal coordinate singularity and to the choice of any other chart on the extension.

The radial-null regularity result now admits a chart-independent statement.

\begin{proposition}[Affine tidal asymptotics]\label{prop:tidal-asymptotics}
Consider the spatially flat, expanding McVittie solution sourced by a positive cosmological constant and a barotropic perfect fluid with constant $w>-1$.  Suppose $0<h_\infty^2<1/27$, so that the limiting black-hole root $\alpha\in(2,3)$ is simple, and define $p$ by Eq.~\eqref{eq:p-def}.  Along every ingoing radial null geodesic ending at the future event horizon \eqref{eq:event-horizon}, the full parallelly propagated Riemann tensor is unbounded if and only if $0<p<2$.  Its divergent component is given by Eq.~\eqref{eq:pp-leading}, and every $C^2$ Lorentzian extension containing that finite-affine endpoint is obstructed.  At $p=2$ the component has a finite nonzero limit, while for $p>2$ its nonstationary part tends to zero and all five independent radial-null frame components remain finite.
\end{proposition}

Section~\ref{sec:extension} expresses the same threshold in coordinates adapted
to a compact family of ingoing characteristics and constructs the two-sided
extension in the finite regime.

\section{Gaussian-null extension across the horizon}
\label{sec:extension}

We now pass from the curvature obstruction along one geodesic to a
neighborhood construction.  The neighborhood condition is essential: a
two-sided metric requires control of mixed longitudinal and transverse
derivatives.  An expansion along a single curve supplies too little data.

\begin{definition}[Anchored local extension]
\label{def:anchored}
Fix an affinely parametrized ingoing radial null geodesic $\gamma$ with finite
future endpoint parameter.  An anchored local $C^k$ extension consists of a
$C^k$ Lorentz metric $\widetilde g$ and a smooth isometric open embedding of a
one-sided neighborhood of $\gamma$ into a spacetime containing its affine
endpoint.  The neighborhood contains a compact interval of neighboring radial
null geodesics, and their endpoints form a local hypersurface $\cH$.
The attached collar is local and may carry freely chosen geometric data.
\end{definition}

\subsection{A parameterized family of characteristics}

Write $X=x-\alpha$ and introduce
\begin{equation}
 z=e^{-\bar\kappa\tau},\qquad
 h(z)=h_\infty\frac{1+z^p}{1-z^p}.
 \label{eq:z-h-exact}
\end{equation}
The ingoing equation \eqref{eq:ingoing-null-slope} is
$x_\tau=F(h(\tau),x)$ with
\begin{equation}
 F(h,x)=-S(x)+hx\sqrt{S(x)}.
 \label{eq:characteristic-F}
\end{equation}
Near $(X,z)=(0,0)$ it becomes the exact regular-singular equation
\begin{align}
 z\partial_zX&=\Phi(X,z^p),\nonumber\\
 \Phi(X,y)&=-\frac{1}{\bar\kappa}
 F\left(h_\infty\frac{1+y}{1-y},\alpha+X\right),
 \label{eq:exact-characteristic}
\end{align}
where the arguments of $F$ denote its Hubble and radial entries.
The function $\Phi$ is analytic near the origin and satisfies
\begin{equation}
 \Phi(0,0)=0,\qquad \Phi_X(0,0)=1.
 \label{eq:phi-linear}
\end{equation}

Fix one horizon-reaching generator.  Stability of the simple root provides a
late-time tube of neighboring characteristics which all converge to $\alpha$.
If $u$ labels their data on a late slice, the first variation satisfies
\begin{align}
 \partial_\tau x_u&=F_xx_u,\nonumber\\
 x_u(\tau,u)&=\exp\left(\int_{\tau_0}^{\tau}
 F_x(\sigma,x(\sigma,u))\,\dd\sigma\right)>0.
 \label{eq:first-variation}
\end{align}
Since $F_x+\bar\kappa$ is integrable along the tube,
$e^{\bar\kappa\tau}x_u$ has a finite positive limit.  The same variation
equations through fourth order show that
\begin{equation}
 v=\lim_{\tau\to\infty}e^{\bar\kappa\tau}X(\tau,u)
 \label{eq:amplitude-coordinate}
\end{equation}
is a $C^4$ coordinate on a compact interval $I_v\Subset(0,\infty)$.
The detailed uniform estimates are given in Appendix~\ref{app:asymptotics}.

Expansion of Eq.~\eqref{eq:exact-characteristic} gives
\begin{equation}
 zX_z=X+qX^2+cz^p+
 \order(X^3,Xz^p,z^{2p}),
 \label{eq:characteristic-quadratic}
\end{equation}
where
\begin{equation}
 q=\frac{4\alpha-9}{2\alpha(\alpha-3)(\alpha-2)},
 \qquad
 c=\frac{2\alpha(\alpha-2)}{\alpha-3}\ne0.
 \label{eq:q-c}
\end{equation}
Putting $X=zW$ removes the sole indicial root and produces
\begin{align}
 W_z&=qW^2+cz^{p-2}
 +\order(zW^3,z^{p-1}W,z^{2p-2}),\nonumber\\
 W(v,0)&=v.
 \label{eq:W-equation}
\end{align}
For $p\ge2$ this is a regular parameter-dependent integral equation.  In
particular,
\begin{align}
 p=2:\quad X={}&vz+(qv^2+c)z^2+\order_{C_v^4C_z^2}(z^3),
 \label{eq:X-p2}\\
 2<p<3:\quad X={}&vz+qv^2z^2+\frac{c}{p-1}z^p
 +\order_{C_v^4C_z^2}(z^3),
 \label{eq:X-p23}\\
 p\ge3:\quad X={}&vz+qv^2z^2
 +\order_{C_v^4C_z^2}(z^3).
 \label{eq:X-p3}
\end{align}
At $p=2$ the right-hand side of Eq.~\eqref{eq:W-equation} is analytic at
$z=0$; the quadratic term is therefore ordinary and no $z^2\log z$ term is
generated.

\subsection{Affine Gaussian-null coordinates}

Let $J=X_v$.  On the physical half-neighborhood define
\begin{equation}
 \rho(v,z)=\frac{1}{\bar\kappa}
 \int_0^z\frac{J(v,\zeta)}{\zeta}\,\dd\zeta.
 \label{eq:rho-integral}
\end{equation}
The amplitude normalization gives $J/z\to1$, hence
\begin{align}
 \rho_z&=\frac{J}{\bar\kappa z}>0,\qquad
 \rho_\tau=-J<0,\nonumber\\
 \rho&=\frac{z}{\bar\kappa}+\frac{qv}{\bar\kappa}z^2
 +\order_{C^2}(z^{\min(3,p)}).
 \label{eq:rho-expansion}
\end{align}
Thus $z=z(v,\rho)$ exists on a closed half-collar and
\begin{equation}
 z=\bar\kappa\rho-qv\bar\kappa^2\rho^2
 +\order_{C^2}(\rho^{\min(3,p)}).
 \label{eq:z-inverse}
\end{equation}

The transformation is exact on the physical side.  In $(\tau,v)$ coordinates
the orbit metric is
\begin{equation}
 g_2=-2J\,\dd\tau\dd v+\frac{J^2}{S}\,\dd v^2,
 \label{eq:tau-v-metric}
\end{equation}
while Eq.~\eqref{eq:rho-integral} gives
$\dd\rho=-J\dd\tau+\rho_v\dd v$.  Consequently
\begin{align}
 \frac{\dd s^2}{m^2}
 &=2\,\dd v\dd\rho+\cF(v,\rho)\dd v^2
 +R(v,\rho)^2\dd\Omega^2,
 \label{eq:gaussian-null}\\
 R(v,\rho)&=\alpha+X(v,z(v,\rho)),
 \nonumber\\
 \cF(v,\rho)&=\frac{J^2}{S}-2\rho_v.
 \label{eq:gn-coefficients}
\end{align}
The vector $\partial_\rho$ is null and affine because
$\Gamma^a{}_{\rho\rho}=0$.  Future approach to the endpoint corresponds to
$\rho\downarrow0$.

The parameter estimates leading to Eqs.~\eqref{eq:X-p2}--\eqref{eq:z-inverse}
show that $R$ and $\cF$ belong to $C^2(I_v\times[0,\epsilon))$.  Their boundary
two-jet is
\begin{align}
 R(v,0)&=\alpha,\qquad
 R_\rho(v,0)=\bar\kappa v,\nonumber\\
 R_{v\rho}(v,0)&=\bar\kappa,
 \label{eq:R-boundary-jet}\\
 R_{\rho\rho}(v,0)&=
 \begin{cases}
 4(\alpha-3)(\alpha-2)/\alpha^3,&p=2,\\
 0,&p>2,
 \end{cases}
 \label{eq:R-rhorho-boundary}\\
 \cF(v,0)&=0,\qquad \cF_\rho(v,0)=0,\nonumber\\
 \cF_{\rho\rho}(v,0)&=2/\alpha^2.
 \label{eq:F-boundary-jet}
\end{align}
The cancellation of the two static quadratic terms in
$R(v,\rho)$ is responsible for the simple form of
Eq.~\eqref{eq:R-rhorho-boundary}.  The warped-product identity gives
\begin{equation}
 R_{\rho A\rho B}=-\frac{R_{\rho\rho}}{R}\delta_{AB},
 \label{eq:gn-warped-tidal}
\end{equation}
which reproduces the finite nonzero limit of
Eq.~\eqref{eq:pp-leading} at $p=2$.

\subsection{The two-sided \texorpdfstring{$C^2$}{C2} collar}

For a $C^2$ function $f$ on $\rho\ge0$, define on
$-\epsilon/3<\rho<0$
\begin{equation}
 (E_2f)(v,\rho)=6f(v,-\rho)-8f(v,-2\rho)+3f(v,-3\rho).
 \label{eq:E2}
\end{equation}
The identities
\begin{equation}
 6-8+3=1,\qquad -6+16-9=1,\qquad 6-32+27=1
 \label{eq:E2-moments}
\end{equation}
match normal derivatives of orders zero, one, and two.  Mixed derivatives
with total order at most two therefore have the same limit from both sides.
Apply $E_2$ separately to $R$ and $\cF$, retaining the original functions for
$\rho\ge0$, and set
\begin{equation}
 \widetilde g=2\,\dd v\dd\rho+(E_2\cF)\dd v^2
 +(E_2R)^2\dd\Omega^2.
 \label{eq:c2-extended-metric}
\end{equation}
After shrinking the collar, $E_2R>\alpha/2$.  The orbit determinant is
identically $-1$, so $\widetilde g$ is Lorentzian.  A smooth timelike field is
given by
\begin{equation}
 T=\partial_v-\frac{1+E_2\cF}{2}\partial_\rho,
 \qquad \widetilde g(T,T)=-1.
 \label{eq:collar-timelike}
\end{equation}

The positive collar is identified with its Gaussian-null image in the
original spacetime.  A smaller nested collar separates the newly attached
points from the remaining original charts, yielding a Hausdorff,
second-countable spacetime with a smooth isometric open embedding.  Every
generator in the selected compact family acquires its endpoint at $\rho=0$.
We have proved the constructive half of the $C^2$ classification.

\begin{proposition}[Local $C^2$ extension]
\label{prop:c2-sufficiency}
Under the assumptions of Sec.~\ref{sec:spacetime}, every horizon-reaching
ingoing radial null generator has an anchored local $C^2$ Lorentzian extension
when $p\ge2$.  The extension contains a compact family of neighboring
generators and has continuous curvature across $\cH$.
\end{proposition}

\section{Optimal H\"older regularity and analytic islands}
\label{sec:regularity}

The $C^2$ collar uses only the first three boundary jets.  The exact
characteristic equation contains considerably more structure: its exponents
form a locally finite additive family, and that family determines the complete
regularity hierarchy.

\begin{definition}[Endpoint classes]
\label{def:holder-extension}
For $N\ge2$ and $0<\vartheta\le1$, an anchored
$C^{N,\vartheta}$ extension is an extension in the sense of
Definition~\ref{def:anchored} whose metric coefficients belong to the standard
big H\"older class $C^{N,\vartheta}$ in every smooth endpoint chart.  At
$\vartheta=1$, derivatives of order $N$ are locally Lipschitz.  An anchored
real-analytic extension carries a compatible real-analytic manifold structure
and has jointly real-analytic metric coefficients on a two-sided endpoint
chart.
\end{definition}

\begin{theorem}[Optimal extension regularity]
\label{thm:main}
Consider the exact spatially flat, expanding $\Lambda+w$ McVittie spacetime
with $w>-1$ and $0<h_\infty^2<1/27$.  Let $\alpha\in(2,3)$ be the simple
limiting black-hole root and define
\begin{equation}
 p=\frac{3(1+w)h_\infty}{\bar\kappa},
 \qquad \bar\kappa=\frac{3-\alpha}{\alpha^2}.
 \label{eq:main-p}
\end{equation}
Fix a horizon-reaching ingoing radial null generator and a compact interval of
neighboring generators.
\begin{enumerate}
\item If $0<p<2$, no anchored local $C^2$ Lorentzian extension exists.
\item If $p=N+\vartheta$ is nonintegral, with $N\ge2$ and
$0<\vartheta<1$, an anchored local $C^{N,\vartheta}$ extension exists.  For
every $\vartheta<\vartheta'\le1$, no anchored
$C^{N,\vartheta'}$ extension exists.
\item If $p\in\mathbb N$ and $p\ge2$, an anchored local real-analytic
Lorentzian extension exists.
\end{enumerate}
The asserted extensions are local and may be chosen on arbitrarily small
collars of the anchored family.
\end{theorem}

\subsection{Polyhomogeneous Gaussian-null germ}

Rewrite Eq.~\eqref{eq:exact-characteristic} as
\begin{equation}
 (z\partial_z-1)X=\Phi(X,z^p)-X,
 \qquad z^{-1}X\big|_{z=0}=v.
 \label{eq:regular-singular-hierarchy}
\end{equation}
For fixed $p\ge2$, define the locally finite index set
\begin{equation}
 E_p=\{j+np:(j,n)\in\mathbb N_0^2,\ (j,n)\ne(0,0)\}.
 \label{eq:Ep}
\end{equation}
We write $f\in\cA^{E_p}$ when, for every cutoff $L\notin E_p$ and all finite
$a,b$, a finite sum $f_{<L}$ of terms with exponent below $L$ satisfies
\begin{equation}
 \sup_{v\in I_v}
 \left|\partial_v^a(z\partial_z)^b(f-f_{<L})\right|
 \le C_{a,b,L}z^L.
 \label{eq:conormal-remainder}
\end{equation}

\begin{proposition}[Characteristic expansion]
\label{prop:polyhom}
For every real $p\ge2$ and compact $I_v$, the solution of
Eq.~\eqref{eq:regular-singular-hierarchy} belongs to $\cA^{E_p}$.  Its
coefficients are uniquely determined by the analytic germ $\Phi$, and its
expansion contains no logarithms.  The conclusion propagates to
$J$, $\rho$, $z(v,\rho)$, $R$, and $\cF$, with ordinary transverse
derivatives lowering the remainder power by their order.
\end{proposition}

\begin{proof}
The first estimate follows from the integral equation for $W=X/z$ and gives
$X-vz=\order(z^2)$.  Suppose all terms below
$\mu\in E_p$ have been removed.  Analyticity of $\Phi$ and closure of $E_p$
under addition identify the leading residual as $f_\mu(v)z^\mu$.  Every forced
exponent satisfies $\mu>1$, and the regular-singular inverse is
\begin{equation}
 z\int_0^z\zeta^{-2}f_\mu(v)\zeta^\mu\,\dd\zeta
 =\frac{f_\mu(v)}{\mu-1}z^\mu.
 \label{eq:volterra-inverse}
\end{equation}
Subtracting this term raises the residual weight.  On a smaller collar the
remaining Volterra map is a contraction in the norm
$\sup z^{-L}|\cdot|$; differentiation in $v$ and $z\partial_z$ gives a
triangular family of the same estimates.  Local finiteness of $E_p$ makes the
induction finite below each cutoff.

The only zero of the divisor $\mu-1$ is the free mode $vz$.  Since no forced
term reaches that exponent, collisions between $j+np$ values add coefficients
without generating logarithms.  Differentiation in $v$ gives $J$,
Eq.~\eqref{eq:rho-integral} preserves the exponents, and the inverse map is
obtained recursively from its nonzero linear coefficient.  Multiplication and
analytic composition preserve $E_p$, proving the statement for $R$ and
$\cF$.  Ordinary derivatives follow from
$\partial_z=z^{-1}(z\partial_z)$.
\end{proof}

For nonintegral $p=N+\vartheta$, every exponent at most $N$ is an integer and
the first nonintegral contribution has power $\rho^p$.  Proposition
\ref{prop:polyhom} therefore gives
\begin{equation}
 R,\cF\in C^{N,\vartheta}
 \bigl(I_v\times[0,\epsilon)\bigr).
 \label{eq:one-sided-holder}
\end{equation}

\subsection{Finite H\"older extension}

For $f\in C^{N,\vartheta}$ on the physical half-collar, set
\begin{equation}
 E_Nf(v,\rho)=
 \begin{cases}
 f(v,\rho),&\rho\ge0,\\[2mm]
 \displaystyle\sum_{j=1}^{N+1}a_jf(v,-j\rho),&\rho<0,
 \end{cases}
 \label{eq:EN}
\end{equation}
Here
\begin{equation}
 a_j=(-1)^{j-1}j\binom{N+2}{j+1}.
 \label{eq:EN-coefficients}
\end{equation}
These are the Lagrange weights which extrapolate a polynomial of
degree at most $N$ from the nodes $-1,\ldots,-N-1$ to $1$.  Hence
\begin{equation}
 \sum_{j=1}^{N+1}a_j(-j)^b=1,
 \qquad 0\le b\le N.
 \label{eq:EN-moments}
\end{equation}
All mixed jets of total order at most $N$ match across $\rho=0$.  On either
side the order-$N$ H\"older seminorm is multiplied by at most
\begin{equation}
 C_{N,\vartheta}=
 \sum_{j=1}^{N+1}|a_j|j^{N+\vartheta}.
 \label{eq:EN-bound}
\end{equation}
For points on opposite sides, subtract the common order-$N$ boundary jet and
apply the one-sided estimates.  This proves boundedness of $E_N$ on the joint
big $C^{N,\vartheta}$ space; Appendix~\ref{app:holder-extension} supplies the
mixed-variable estimate explicitly.

Applying $E_N$ to $R$ and $\cF$ in Eq.~\eqref{eq:gaussian-null} gives
\begin{equation}
 \widetilde g=2\,\dd v\dd\rho+(E_N\cF)\dd v^2
 +(E_NR)^2\dd\Omega^2.
 \label{eq:holder-metric}
\end{equation}
For negative width smaller than $\epsilon/(N+1)$ all sampled points remain in
the physical collar.  After a further shrink, $E_NR>\alpha/2$, while the
orbit determinant remains $-1$.  The nested-collar gluing used in
Sec.~\ref{sec:extension} then yields an anchored
$C^{N,\vartheta}$ Lorentzian extension.

\subsection{Integer exponents}

Let $p\ge2$ be an integer.  The exact equation for $W=X/z$ is
\begin{equation}
 W_z=\frac{\Phi(zW,z^p)-zW}{z^2}.
 \label{eq:analytic-W}
\end{equation}
The numerator vanishes to second order at $z=0$, and the quotient is jointly
analytic in $(v,z,W)$.  Analytic dependence for ordinary differential
equations gives $W(v,z)$ analytic near the endpoint.  Since
$J/z=W_v$ is analytic, Eq.~\eqref{eq:rho-integral} gives an analytic
$\rho(v,z)$ with $\rho_z(v,0)=1/\bar\kappa$.  The analytic inverse-function
theorem then yields $z(v,\rho)$ and jointly analytic functions $R$ and $\cF$.
Their convergent endpoint series define the same coefficients for positive
and negative $\rho$.  On a sufficiently small two-sided collar they therefore
give a real-analytic Lorentz metric.  Integer collisions in $E_p$ merely
combine ordinary Taylor coefficients, explaining the analytic islands.

\subsection{Invariant sharpness}

Along $v=\mathrm{const}$ take $k=\partial_\rho$ and a parallelly propagated
unit angular vector.  The warped-product identity defines the tidal scalar
\begin{equation}
 \cT(\rho)=R_{kAkA}
 =\frac12\operatorname{Ric}(k,k)
 =-\frac{R_{\rho\rho}}{R}.
 \label{eq:tidal-scalar}
\end{equation}
The first noninteger radius coefficient obtained from
Proposition~\ref{prop:polyhom} is
\begin{equation}
 R=R_{\rm int}
 +\frac{c(\alpha)\bar\kappa^p}{p-1}\rho^p
 +\order_{\mathrm b}(\rho^{p+1}),
 \label{eq:R-leading-noninteger}
\end{equation}
where $R_{\rm int}$ contains the finitely many relevant integer powers and
$c(\alpha)$ is given in Eq.~\eqref{eq:q-c}.  Thus
\begin{align}
 \cT(\rho)&=P(\rho)+K(\alpha,p)\rho^{p-2}
 +\order_{\mathrm b}(\rho^{p-1}),
 \label{eq:T-expansion}\\
 K(\alpha,p)&=-\frac{p c(\alpha)\bar\kappa^p}{\alpha}
 =\frac{2p(\alpha-2)\bar\kappa^{p-1}}{\alpha^2}>0.
 \label{eq:K-coefficient}
\end{align}
Here $P$ is a finite Taylor polynomial.  The second equality in
Eq.~\eqref{eq:K-coefficient} follows from the root relations and also follows
directly from $-\mathcal Q^2\sqrt S\,\dot h$ with
$\mathcal Q=-J^{-1}$ in the present affine normalization.

For $p=N+\vartheta$ and $m=N-2$, Eq.~\eqref{eq:T-expansion} gives
\begin{equation}
 \partial_\rho^m\cT
 =\partial_\rho^mP+
 K\frac{\Gamma(p-1)}{\Gamma(\vartheta+1)}\rho^\vartheta
 +o(\rho^\vartheta).
 \label{eq:T-holder-obstruction}
\end{equation}
The nonzero power $\rho^\vartheta$ belongs to no
$C^{0,\vartheta'}$ with $\vartheta'>\vartheta$.  If the metric admitted an
anchored $C^{N,\vartheta'}$ extension, its Riemann tensor would be
$C^{N-2,\vartheta'}$.  The geodesic and parallel-frame transport equations,
whose coefficients are the Christoffel symbols, preserve the regularity
required for the contraction \eqref{eq:tidal-scalar}.  Its restriction to the
anchored generator would then lie in $C^{N-2,\vartheta'}$, contradicting
Eq.~\eqref{eq:T-holder-obstruction}.  A smooth endpoint change has
$\rho=a\widehat\rho+\order(\widehat\rho^2)$ with $a\ne0$ and only multiplies
the leading coefficient by a nonzero factor.  The obstruction is therefore
independent of the endpoint chart and affine normalization.

\begin{proof}[Proof of Theorem~\ref{thm:main}]
For $p<2$, Proposition~\ref{prop:tidal-asymptotics} gives an unbounded p.p.
Riemann component at a finite-affine endpoint, which is incompatible with a
$C^2$ extension.  For nonintegral $p\ge2$, Eqs.~\eqref{eq:one-sided-holder}
and \eqref{eq:holder-metric} construct the attained
$C^{N,\vartheta}$ collar, while Eq.~\eqref{eq:T-holder-obstruction} excludes
all higher H\"older exponents.  Equation~\eqref{eq:analytic-W} gives the
real-analytic collar for every integer $p\ge2$.
\end{proof}

\begin{corollary}[Integer differentiability classes]
\label{cor:ck-classification}
For every integer $k\ge2$, an anchored local geometric $C^k$ extension exists
if and only if
\begin{equation}
 p\ge k
 \qquad\text{or}\qquad
 p\in\mathbb N,\quad p\ge2.
 \label{eq:ck-classification}
\end{equation}
\end{corollary}

\section{Matter limits and physical consequences}
\label{sec:consequences}

\subsection{Critical Einstein tensor and perfect-fluid open sides}

For the spherically symmetric Gaussian-null metric
\begin{equation}
 g=2\,\dd v\dd\rho+\cF\dd v^2+R^2\dd\Omega^2,
 \label{eq:gn-matter-metric}
\end{equation}
the radial Einstein equation contains
\begin{equation}
 G_{\rho\rho}=-\frac{2R_{\rho\rho}}{R}.
 \label{eq:G-rhorho}
\end{equation}
Substitution of the boundary two-jet
\eqref{eq:R-boundary-jet}--\eqref{eq:F-boundary-jet} gives, with
$\Lambda=3h_\infty^2$,
\begin{equation}
 G^a{}_b\big|_{\cH}
 =-\Lambda\delta^a{}_b+\nu
 (\partial_v)^a(\dd\rho)_b,
 \label{eq:G-boundary}
\end{equation}
where
\begin{equation}
 \nu=
 \begin{cases}
 \displaystyle\frac{8(3-\alpha)(\alpha-2)}{\alpha^4}>0,&p=2,\\[2mm]
 0,&p>2.
 \end{cases}
 \label{eq:nu-boundary}
\end{equation}
At the critical exponent the second term is rank one, nilpotent, and nonzero.
The boundary Einstein endomorphism is therefore of type II.  A perfect-fluid
endomorphism defined by a finite unit timelike velocity is diagonalizable,
apart from the vacuum-energy degeneracy where it is proportional to the
identity.  The type-II value at $p=2$ has no such boundary decomposition.
This affine boundary limit retains the finite effect of the decaying matter
after its boost by the ingoing null frame.

There is also a two-sided completion with exact perfect-fluid regions away
from the join.  Take a second McVittie half-neighborhood with the same
parameters and use its characteristic amplitude $\widehat v$ and affine
coordinate $\widehat\rho$.  The identification
\begin{equation}
 (\widehat v,\widehat\rho)=(-v,-\rho)
 \label{eq:paired-identification}
\end{equation}
preserves $2\dd v\dd\rho$ and matches every Gaussian-null jet of total order
at most two.  Each open side satisfies the original McVittie
Einstein--perfect-fluid equations.  The joined $C^2$ metric has a continuous
Einstein tensor, so $T_{ab}=G_{ab}/(8\pi)$ is continuous, distributionally
conserved, and carries no delta-function curvature or null shell.  Appendix
\ref{app:perfect-fluid} gives the matching and weak Bianchi argument.

The natural future tangent of the ingoing family is $-\partial_\rho$ on the
first copy and $-\partial_{\widehat\rho}$ on the second.  Under
Eq.~\eqref{eq:paired-identification} these tangents have opposite images.  The
collar admits the continuous time orientation generated by
Eq.~\eqref{eq:collar-timelike}; that orientation agrees with the original
expanding orientation on one of the two copies.  A join preserving both
expanding orientations would require a different identification.

\subsection{Equation-of-state form of the hierarchy}

At fixed limiting root, the causal resonance and the $C^2$ threshold occur at
\begin{equation}
 w_1=\frac{\bar\kappa}{3h_\infty}-1,
 \qquad
 w_2=\frac{2\bar\kappa}{3h_\infty}-1.
 \label{eq:w-thresholds}
\end{equation}
The first is the $p=1$ borderline in the causal classification of
Ref.~\cite{DaSilvaFontaniniGuariento2013}; the second is the boundary between
$C^2$ inextendibility and the extension hierarchy of
Theorem~\ref{thm:main}.  More generally, the line $p=k$ for integer $k\ge2$
is
\begin{equation}
 w_k=\frac{k\bar\kappa}{3h_\infty}-1.
 \label{eq:wk-threshold}
\end{equation}
Nonintegral points between these lines carry their exact H\"older exponent,
whereas the equality points with integer $p$ lie on analytic islands.

For $h_\infty=0.1$, one finds
$\alpha\simeq2.091488$ and
\begin{equation}
 w_1\simeq-0.307694,
 \qquad w_2\simeq0.384612.
 \label{eq:default-thresholds}
\end{equation}
Dust and radiation then have $p<2$, while a stiff fluid has $p>2$.
Increasing $w$ accelerates cosmological dilution without changing the
limiting surface gravity, moving the endpoint upward through the regularity
hierarchy.

\begin{figure*}[t]
 \includegraphics[width=0.88\textwidth]{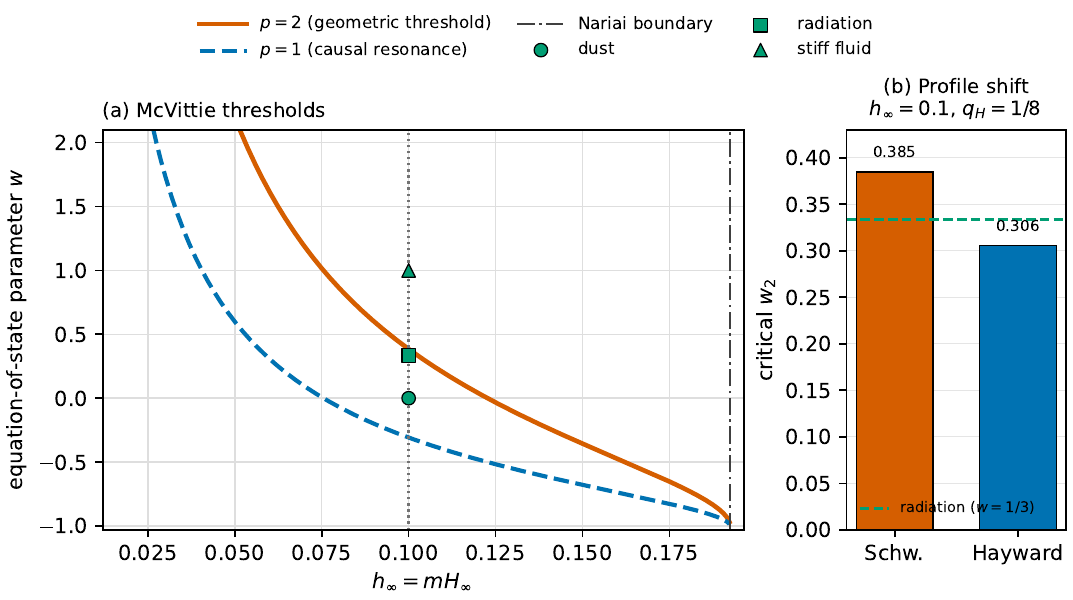}
 \caption{Decay thresholds at the limiting McVittie black-hole horizon.
 Panel (a) shows the causal resonance $p=1$ and the geometric threshold
 $p=2$ throughout the simple-root domain
 $0<h_\infty<1/\sqrt{27}$.  The vertical line marks $h_\infty=0.1$, with
 dust, radiation, and stiff-fluid values indicated.  Panel (b) compares the
 $p=2$ equation-of-state threshold for the Schwarzschild and Hayward limiting
 profiles at $h_\infty=0.1$ and $q_H=1/8$.}
 \label{fig:thresholds}
\end{figure*}

A finite collection of noninteracting fluids with positive asymptotic
densities is controlled by the slowest nonvanishing component.  Equal decay
exponents add with a positive coefficient, and faster components enter only at
higher orders.  Appendix~\ref{app:profiles} derives this rule and also records
the Hayward comparison shown in Fig.~\ref{fig:thresholds}.  The latter
illustrates how changing the limiting root slope moves the same cosmological
tail across the $p=2$ threshold.

\subsection{Integrated tidal strength and Jacobi fields}

The leading transverse tidal field has magnitude $s^{p-2}$.  Its Kr\'olak and
Tipler integrals behave as
\begin{align}
 I_K(s)&=\int_s^{s_*}u^{p-2}\,\dd u,
 \label{eq:krolak-integral}\\
 I_T(s)&=\int_s^{s_*}\dd u\int_u^{s_*}v^{p-2}\,\dd v.
 \label{eq:tipler-integral}
\end{align}
The first diverges for $0<p<1$ and logarithmically at $p=1$, while the second
is finite for all $p>0$ \cite{Tipler1977,ClarkeKrolak1985,Krolak1986,Ori2000}.
The angular Jacobi equation has the form
\begin{equation}
 J''+Cs^{p-2}J=\text{higher-order terms}.
 \label{eq:jacobi-equation}
\end{equation}
For $p\ne1$, a solution with $J(0)=J_0\ne0$ satisfies
\begin{equation}
 J(s)=J_0+J_1s-\frac{CJ_0}{p(p-1)}s^p+o(s^p),
 \label{eq:jacobi-expansion}
\end{equation}
and at $p=1$ the last term is $-CJ_0s\log s$.  Two independent angular
Jacobi fields retain finite nonzero lengths and a finite area element.  The
range $1<p<2$ gives a particularly clear distinction: accumulated tidal
distortion is finite, while the instantaneous curvature still excludes a
$C^2$ completion.

\begin{table*}[t]
\caption{Affine-endpoint classification for the exact $\Lambda+w$ McVittie
family.  Here $N=\lfloor p\rfloor$, $\vartheta=p-N$, and K/T denote the
Kr\'olak and Tipler integrals.}
\label{tab:classification}
\begin{ruledtabular}
\small
\setlength{\tabcolsep}{2.5pt}
\begin{tabular}{@{}lcccc@{}}
$p$ & Characteristic & Regularity & Tidal term & Strength\\
\hline
$0<p<1$ & subcritical & no $C^2$ & $s^{p-2}$ diverges & K diverges; T finite\\
$p=1$ & $s\log s$ & no $C^2$ & $s^{-1}$ diverges & K logarithmic; T finite\\
$1<p<2$ & supercritical & no $C^2$ & $s^{p-2}$ diverges & K/T finite\\
$p=2$ & quadratic & analytic & finite; type II & K/T finite\\
nonint. $p>2$ & polyhomogeneous & $C^{N,\vartheta}$ & $s^{p-2}$ tail & K/T finite\\
integer $p\ge3$ & Taylor & analytic & smooth tail & K/T finite\\
\end{tabular}
\end{ruledtabular}
\end{table*}

\section{Discussion}
\label{sec:discussion}

The future McVittie black-hole horizon is governed by two exponential scales.
The barotropic component decays at rate $3(1+w)H_\infty$, and the limiting
simple root converts cosmological time to affine distance at rate $\kappa$.
Their ratio $p$ controls the transverse powers of the compactified metric.
The causal resonance occurs at $p=1$, where the radial characteristic develops
an $s\log s$ term.  Curvature introduces two affine derivatives and places the
geometric extension threshold at $p=2$.

The classification follows from two complementary arguments.  The
parallelly propagated null frame gives an invariant obstruction attached to a
specified finite-affine endpoint.  The parameterized characteristic family
supplies the mixed estimates needed to construct a full neighborhood of that
endpoint.  Together they prove a $C^2$ extension precisely for $p\ge2$ and
identify the optimal class at every larger nonintegral exponent.  The
integer values form analytic islands because the exact coth tail and the
characteristic equation become ordinary analytic functions of the affine
horizon variable.  This arithmetic distinction would be invisible in a
classification based only on curvature boundedness.

The exponent $p-2$ already appears in the late-time analysis of Kaloper,
Kleban, and Martin \cite{KaloperKlebanMartin2010}.  The present result assigns
it a tensorial affine meaning, proves the required parameter-uniform
Gaussian-null estimates, and constructs the attached Lorentzian collar.
Lake and Abdelqader's exact $\Lambda$CDM extension
\cite{LakeAbdelqader2011} and the causal theorem of da Silva, Fontanini, and
Guariento \cite{DaSilvaFontaniniGuariento2013} provide the closest global and
causal comparisons.  Nolan's radial-null endpoint theorem
\cite{Nolan2014,Nolan2026} supplies the global event-horizon bridge used here,
while the finite timelike p.p. behavior found in Ref.~\cite{Nolan2026}
reflects the different boost carried by timelike and radial-null frames.

At the critical value $p=2$, the geometric extension is analytic even though
the affine boundary Einstein tensor has a type-II nilpotent part.  Two exact
McVittie perfect-fluid regions can be joined at $C^2$ regularity, with a
continuous conserved stress tensor and no surface layer.  Their natural
expanding time orientations meet with opposite signs in this construction.
The boundary algebraic type, the equations on the two open sides, and the
orientation of the join are therefore three distinct pieces of geometric
information.

The classification is local and anchored.  It leaves $C^0$ and $C^1$
extensions, uniqueness, maximality, and dynamically selected data on the
attached side open.  These questions involve low-regularity causal structure
and characteristic evolution beyond the endpoint.  Within the classical
curvature range $k\ge2$, Theorem~\ref{thm:main} and
Corollary~\ref{cor:ck-classification} give the complete hierarchy for the exact
$\Lambda+w$ McVittie family.

The mechanism extends naturally to other asymptotically stationary horizons.
An exterior mode $e^{-\lambda t}$ is converted by redshift into the affine
power $\rho^{\lambda/\kappa}$; collisions, logarithms, and complex frequencies
then determine the endpoint class.  McVittie isolates this principle in an
exact nonlinear spacetime and shows that the regularity of a cosmological
black-hole horizon records both cosmological dilution and local horizon
redshift.

\appendix
\section{McVittie geometry and exact background}
\label{app:geometry}

For reference, the nonangular part of Eq.~\eqref{eq:mcvittie-areal} and its inverse are
\begin{align}
 g_{AB}&=\begin{pmatrix}
 -f & -hx/\sqrt S\\
 -hx/\sqrt S & 1/S
 \end{pmatrix},\nonumber\\
 g^{AB}&=\begin{pmatrix}
 -1/S & -hx/\sqrt S\\
 -hx/\sqrt S & f
 \end{pmatrix},
 \label{eq:orbit-metric-inverse}
\end{align}
where $A,B\in\{\tau,x\}$.  Both determinants equal $-1$.  Since $x$ is the dimensionless areal radius,
\begin{equation}
 g^{ab}\nabla_ax\nabla_bx=f,
 \label{eq:gradient-x}
\end{equation}
which gives Eq.~\eqref{eq:f-infinity} in the asymptotic limit.  The transformation from Eq.~\eqref{eq:mcvittie-isotropic} follows from
\begin{align}
 \sqrt{1-\frac{2m}{R}}&=\frac{1-\mu}{1+\mu},\nonumber\\
 \dd R&=HR\sqrt{1-\frac{2m}{R}}\,\dd t
 +a(1-\mu^2)\,\dd r.
 \label{eq:isotropic-to-areal}
\end{align}

Let $u^a=S^{-1/2}(\partial_\tau)^a$ in the isotropic chart, transported to the areal chart.  Direct contraction of the Einstein tensor gives
\begin{align}
 G_{ab}u^au^b&=3h^2,\label{eq:Einstein-rho}\\
 G_{ab}e_{(i)}^ae_{(j)}^b
 &=\left(-3h^2-\frac{2\dot h}{\sqrt S}\right)\delta_{ij},
 \label{eq:Einstein-pressure}\\
 G_{ab}u^ae_{(i)}^b&=0.
 \label{eq:Einstein-flux}
\end{align}
Restoring the factor $m^{-2}$ converts these expressions to Eq.~\eqref{eq:mcvittie-fluid}.  The equality of the three spatial eigenvalues and Eq.~\eqref{eq:Einstein-flux} establish the perfect-fluid form directly from the geometry.  The Friedmann equations enter separately in the evolution of the cosmological source.

For a flat FLRW background containing $\Lambda=3H_\infty^2$ and a barotropic component, the conservation and Friedmann equations are
\begin{equation}
 \dot\rho_w+3H(1+w)\rho_w=0,
 \qquad
 H^2=H_\infty^2+\frac{8\pi}{3}\rho_w.
 \label{eq:friedmann-system}
\end{equation}
They imply
\begin{equation}
 \frac{\dd H}{\dd t}=-4\pi(1+w)\rho_w
 =-\frac{3}{2}(1+w)(H^2-H_\infty^2).
 \label{eq:raychaudhuri-dimensional}
\end{equation}
Integrating with a big-bang origin at $t=0$ gives Eqs.~\eqref{eq:scale-factor} and \eqref{eq:hubble-exact}.  The identity
\begin{equation}
 \coth z-1=\frac{2e^{-2z}}{1-e^{-2z}}
 \label{eq:coth-identity}
\end{equation}
then yields Eqs.~\eqref{eq:ydef}--\eqref{eq:coth-tail} and shows that the tail is differentiable term by term to every finite order.

Finally, multiplying $f_\infty(\alpha)=0$ by $\alpha$ gives
\begin{equation}
 h_\infty^2\alpha^3-\alpha+2=0.
 \label{eq:root-polynomial}
\end{equation}
Its positive stationary point occurs at $\alpha=1/(\sqrt3h_\infty)$ and is a local minimum.  Under condition \eqref{eq:subnariai} this minimum is negative, producing two positive simple zeros.  The smaller zero lies in $(2,3)$, and differentiation followed by Eq.~\eqref{eq:alpha-relation} gives Eq.~\eqref{eq:kappa-bar}.

\section{Affine radial null frame and curvature components}
\label{app:null-frame}

Suppressing the angular directions, the affine Lagrangian is
\begin{equation}
 2L=-f\dot\tau_\ell^2-\frac{2hx}{\sqrt S}\dot\tau_\ell\dot x_\ell
 +\frac{1}{S}\dot x_\ell^2,
 \label{eq:null-lagrangian}
\end{equation}
where a dot with subscript $\ell$ denotes $\dd/\dd\ell$.  The null constraint factorizes to give
\begin{equation}
 \frac{\dd x}{\dd\tau}=\sqrt S(hx\pm\sqrt S).
 \label{eq:both-null-slopes}
\end{equation}
The minus sign is Eq.~\eqref{eq:ingoing-null-slope}; the plus sign gives the outgoing family.  Substitution of the ingoing branch into the Euler--Lagrange equation for $\tau$ gives Eq.~\eqref{eq:Q-equation}.  Conversely, Eqs.~\eqref{eq:ingoing-null-slope} and \eqref{eq:Q-equation}, together with the null constraint, make both Euler--Lagrange residuals vanish.  This supplies a closed first-order affine system for $(\tau,x,\mathcal Q)$.

The vectors in Eqs.~\eqref{eq:k-vector}, \eqref{eq:n-vector}, and \eqref{eq:angular-frame} obey
\begin{align}
 k^2=n^2&=0,& k\cdot n&=-1,\nonumber\\
 e_{\hat A}\cdot e_{\hat B}&=\delta_{AB},&
 k\cdot e_{\hat A}=n\cdot e_{\hat A}&=0.
 \label{eq:null-frame-normalization}
\end{align}
The orbit connection and Eq.~\eqref{eq:Q-equation} give $\nabla_kk=\nabla_kn=0$.  For the angular frame, the only required warped-product connection is
\begin{equation}
 \nabla_k\partial_{\hat A}=\frac{k(x)}{x}\partial_{\hat A},
 \label{eq:warp-connection}
\end{equation}
which is canceled by differentiating the factor $x^{-1}$ in Eq.~\eqref{eq:angular-frame}.  Hence the complete frame is parallelly propagated.

For a metric $g=q_{AB}\dd y^A\dd y^B+x^2\gamma_{ij}\dd z^i\dd z^j$, the mixed orbit--sphere curvature is
\begin{equation}
 R_{AiBj}=-x(\nabla_A\nabla_Bx)\gamma_{ij}.
 \label{eq:warped-riemann}
\end{equation}
Contracting twice with $k$ and using $\nabla_kk=0$ yields Eq.~\eqref{eq:warped-identity}.  Evaluation of $\dd^2x/\dd\ell^2$ with the affine system gives
\begin{equation}
 -\frac{1}{x}\frac{\dd^2x}{\dd\ell^2}
 =-\mathcal Q^2\sqrt S\,\dot h,
 \label{eq:warped-evaluation}
\end{equation}
independently reproducing Eq.~\eqref{eq:RkAkB}.

The orbit curvature, mixed warped curvature, and intrinsic curvature of the symmetry spheres generate all remaining contractions.  In the ordered null frame $(k,n,e_{\hat\theta},e_{\hat\phi})$, the algebraically independent nonzero entries are precisely Eqs.~\eqref{eq:Rknkn}--\eqref{eq:Rangular}; all others follow from
\begin{equation}
 R_{IJKL}=-R_{JIKL}=-R_{IJLK}=R_{KLIJ}
 \label{eq:riemann-symmetries}
\end{equation}
and spherical interchange of $\hat\theta$ and $\hat\phi$.  This list also shows that Eq.~\eqref{eq:RkAkB} is the only component enhanced by $\mathcal Q^2\dot h$; $R_{nAnB}$ is suppressed by the inverse boost, and the remaining components have finite limits.

\section{Uniform characteristic and endpoint estimates}
\label{app:asymptotics}

This appendix supplies the parameter estimates used in
Sec.~\ref{sec:extension}.  Fix a sufficiently late slice $\tau=\tau_0$ and a
compact interval $I_u$ of ingoing initial radii around the anchored generator.
In a small stable tube about $(\infty,\alpha)$,
\begin{equation}
 X_\tau=-\bar\kappa X+\order(X^2+y),
 \qquad y=e^{-\bar\lambda\tau}.
 \label{eq:app-X-stable}
\end{equation}
An exponential barrier first yields
$X=\order(e^{-(\bar\kappa-\eta)\tau})$ for any fixed small $\eta>0$.
Variation of constants and the integrability of the nonlinear remainder then
improve this to
\begin{equation}
 X=\order(e^{-\bar\kappa\tau})+
 \order(e^{-\bar\lambda\tau})
 =\order(e^{-\bar\kappa\tau})
 \label{eq:app-X-decay}
\end{equation}
uniformly on $I_u$ when $p\ge2$.

Let $J_a=\partial_u^a x$.  Differentiating the scalar flow gives
\begin{align}
 \dot J_1={}&F_xJ_1,\nonumber\\
 \dot J_2={}&F_xJ_2+F_{xx}J_1^2,\nonumber\\
 \dot J_3={}&F_xJ_3+3F_{xx}J_1J_2+F_{xxx}J_1^3,\nonumber\\
 \dot J_4={}&F_xJ_4+4F_{xx}J_1J_3+3F_{xx}J_2^2\nonumber\\
 &+6F_{xxx}J_1^2J_2+F_{xxxx}J_1^4.
 \label{eq:app-variations}
\end{align}
Here $F$ is the algebraic velocity in Eq.~\eqref{eq:characteristic-F},
evaluated at $(h(\tau),x(\tau,u))$.  Because
$F_x+\bar\kappa=\order(X+y)$ is integrable, the first equation gives
\begin{equation}
 \lim_{\tau\to\infty}e^{\bar\kappa\tau}J_1(\tau,u)>0
 \label{eq:app-J1-limit}
\end{equation}
uniformly on the compact interval.  The source of each higher variation
contains at least two decaying lower variations.  Multiplication by the
homogeneous integrating factor makes these sources integrable, and induction
gives
\begin{equation}
 \partial_u^aX=\order(e^{-\bar\kappa\tau}),
 \qquad 1\le a\le4.
 \label{eq:app-four-variations}
\end{equation}
It follows that the limit in Eq.~\eqref{eq:amplitude-coordinate} is $C^4$ and
has positive derivative.  Passing from $u$ to $v$ preserves the uniform
estimates.

To control transverse derivatives, write $X=zW$.  On a compact $v$ interval
Eq.~\eqref{eq:W-equation} is a parameter-dependent Volterra equation.  After
subtracting the terms displayed in
Eqs.~\eqref{eq:X-p2}--\eqref{eq:X-p3}, its right-hand side and its first two
$z$ derivatives have the corresponding weighted integrable bounds.
Differentiation up to four times in $v$ produces linear Volterra equations
whose sources are finite products of already controlled variations.
Gronwall's inequality yields
\begin{equation}
 \partial_v^a\partial_z^b\mathcal R(v,z)
 =\order(z^{3-b}),
 \qquad 0\le a\le4,\quad0\le b\le2,
 \label{eq:app-weighted-remainder}
\end{equation}
for the remainder $\mathcal R$ in the three displayed regimes, with the
obvious stronger power whenever the next exponent exceeds three.

Equation~\eqref{eq:rho-integral} integrates $J/z$ without loss of parameter
regularity.  Since $\rho_z(v,0)=1/\bar\kappa$, the parameter-dependent inverse
function theorem gives $z(v,\rho)$ with the same two transverse derivatives.
The derivative budget for $\cF=J^2/S-2\rho_v$ explains the use of four
longitudinal variations: two $v$ derivatives of $\rho_v$ contain the integral
of $X_{vvvv}/z$.  Consequently every mixed derivative of $R$ and $\cF$ of
total order at most two extends continuously to $\rho=0$.

For the all-order statement, let $X_{<L}$ be the finite formal sum determined
recursively below a cutoff $L\notin E_p$.  The remainder $Y=X-X_{<L}$ obeys
\begin{equation}
 (z\partial_z-1)Y=A(v,z)Y+B(v,z,Y)+\order(z^L),
 \label{eq:app-remainder-equation}
\end{equation}
where $A=\order(z)$ and $B=\order(Y^2)$.  The inverse
\eqref{eq:volterra-inverse}, followed by a contraction on a sufficiently small
collar, gives $Y=\order(z^L)$.  Applying $\partial_v^a(z\partial_z)^b$ yields a
triangular system with the same leading operator and proves
Eq.~\eqref{eq:conormal-remainder}.  This argument also shows directly that the
constants may depend on the fixed values of $\alpha$, $p$, $L$, and the
derivative orders while remaining uniform on $I_v$.

\section{Finite reflection in joint H\"older spaces}
\label{app:holder-extension}

Let $Q_+=I_v\times[0,\epsilon)$ and let
$f\in C^{N,\vartheta}(Q_+)$, where $N\ge2$ and
$0<\vartheta<1$.  Choose $\delta<\epsilon/(N+1)$ and define $E_Nf$ by
Eq.~\eqref{eq:EN} on
$Q=I_v\times(-\delta,\epsilon)$.  The moment identities
\eqref{eq:EN-moments} give, for $a+b\le N$,
\begin{equation}
 \lim_{\rho\uparrow0}
 \partial_v^a\partial_\rho^bE_Nf(v,\rho)
 =\partial_v^a\partial_\rho^bf(v,0).
 \label{eq:app-mixed-jet}
\end{equation}
Thus all derivatives through order $N$ extend continuously across the
boundary.

For two points on the negative side, the chain rule and the H\"older estimate
on $Q_+$ give
\begin{equation}
 [D^NE_Nf]_{\vartheta;Q_-}
 \le
 \left(\sum_{j=1}^{N+1}|a_j|j^{N+\vartheta}\right)
 [D^Nf]_{\vartheta;Q_+}.
 \label{eq:app-negative-holder}
\end{equation}
Consider next $P=(v,\rho)$ with $\rho\ge0$ and
$\widehat P=(\widehat v,\widehat\rho)$ with $\widehat\rho<0$.
Let $T_Nf$ be the joint Taylor polynomial at a boundary point between their
$v$ coordinates.  The moment identities imply $E_NT_Nf=T_Nf$.
Subtracting this common polynomial and applying the one-sided Taylor
remainder estimates gives
\begin{equation}
 |D^N(E_Nf)(P)-D^N(E_Nf)(\widehat P)|
 \le C\|f\|_{C^{N,\vartheta}(Q_+)}
 |P-\widehat P|^\vartheta.
 \label{eq:app-cross-holder}
\end{equation}
The constant depends only on $N$, $\vartheta$, the compact $v$ interval, and
the collar widths.  Equations~\eqref{eq:app-mixed-jet}--\eqref{eq:app-cross-holder}
prove
\begin{equation}
 \|E_Nf\|_{C^{N,\vartheta}(Q)}
 \le C_{N,\vartheta,Q}\|f\|_{C^{N,\vartheta}(Q_+)}.
 \label{eq:app-extension-bound}
\end{equation}

The argument applies in each regular angular chart and to every scalar
coefficient of the spherically symmetric metric.  On overlaps the positive
side transition functions are unchanged.  Shrinking the new collar and using
the same finite angular atlas therefore produces a joint
$C^{N,\vartheta}$ Lorentz metric on a neighborhood of the full local horizon
sphere.

\section{Einstein tensor and the exact-copy join}
\label{app:perfect-fluid}

On every physical McVittie open set, the mixed Einstein tensor has
characteristic polynomial
\begin{equation}
 \det(G^a{}_b-\lambda\delta^a{}_b)
 =(\lambda+8\pi\rho)(\lambda-8\pi P)^3.
 \label{eq:Einstein-characteristic}
\end{equation}
Its timelike eigenspace gives the perfect-fluid decomposition, and the radial
energy current in that frame vanishes.  Along the ingoing family,
\begin{equation}
 h^2-h_\infty^2=\frac{4h_\infty^2z^p}{(1-z^p)^2},
 \qquad
 \dot h=-2\bar\kappa p h_\infty
 \frac{z^p}{(1-z^p)^2}.
 \label{eq:app-matter-tail}
\end{equation}
The affine contraction retains the boost factor,
\begin{equation}
 T_{ab}k^ak^b=(\rho+P)\mathcal Q^2S
 \sim C_T\rho^{p-2},
 \label{eq:app-null-matter}
\end{equation}
which is finite and nonzero at $p=2$.

For completeness, the Einstein tensor of
$g=2\dd v\dd\rho+\cF\dd v^2+R^2\dd\Omega^2$ obeys
\begin{align}
 G_{\rho\rho}&=-\frac{2R_{\rho\rho}}R,
 \label{eq:app-Grr}\\
 \frac{G_{\theta\theta}}{R^2}
 &=-\frac{2\cF R_{\rho\rho}+R\cF_{\rho\rho}
 +2\cF_\rho R_\rho-4R_{v\rho}}{2R}.
 \label{eq:app-Gangular}
\end{align}
The remaining orbit entries follow from the same warped-product formulas.
Insertion of Eqs.~\eqref{eq:R-boundary-jet}--\eqref{eq:F-boundary-jet}
gives Eq.~\eqref{eq:G-boundary}.  On the orbit plane its mixed form is
\begin{equation}
 G^i{}_j\big|_{\cH}=
 \begin{pmatrix}
 -\Lambda&\nu\\
 0&-\Lambda
 \end{pmatrix},
 \qquad
 G^A{}_B\big|_{\cH}=-\Lambda\delta^A{}_B.
 \label{eq:app-G-matrix}
\end{equation}
At $p=2$, $(G+\Lambda I)^2=0$ and $G+\Lambda I\ne0$, establishing the
rank-one type-II limit.

We next construct the second exact open side.  On an independent copy choose a
compact negative amplitude interval $-\!I_v$ and solve
\begin{equation}
 W_z=\frac{\Phi(zW,z^p)-zW}{z^2},
 \qquad W(\widehat v,0)=\widehat v.
 \label{eq:app-second-copy}
\end{equation}
For $p\ge2$ the right-hand side extends continuously to $z=0$ and is locally
Lipschitz in $W$.  A sufficiently small uniform interval gives
$x=\alpha+zW>2$ and $W_{\widehat v}>0$, so the affine construction
\eqref{eq:rho-integral} applies.  This produces an actual McVittie
half-neighborhood with Gaussian-null coordinates
$(\widehat v,\widehat\rho)$.

Under $(\widehat v,\widehat\rho)=(-v,-\rho)$, a derivative of order $(a,b)$
acquires the factor $(-1)^{a+b}$.  The boundary data
\begin{align}
 R&=\alpha,\qquad R_\rho=\bar\kappa v,\nonumber\\
 R_{v\rho}&=\bar\kappa,\qquad R_{\rho\rho}=A_p,\nonumber\\
 \cF&=\cF_\rho=0,\qquad\cF_{\rho\rho}=2/\alpha^2
 \label{eq:app-paired-jets}
\end{align}
are invariant under the simultaneous sign reversal, with
$A_p$ given by Eq.~\eqref{eq:R-rhorho-boundary}.  The remaining jets of total
order at most two vanish or follow by differentiating these expressions.
The piecewise metric is therefore jointly $C^2$.

Each open side satisfies the classical Einstein equation.  Since a $C^2$
metric has continuous Einstein tensor, the two stress tensors join uniquely
to $T=G/(8\pi)$.  There is no jump in the first metric jet and hence no
distributional curvature supported on $\cH$.  To verify conservation, choose
smooth Lorentz metrics $g_\epsilon\to g$ in $C^2_{\rm loc}$.  For every compactly
supported test one-form $\varphi_b$, the smooth Bianchi identity gives
\begin{equation}
 \int G(g_\epsilon)^{ab}\nabla^{g_\epsilon}_a\varphi_b\,
 \dd\mu_{g_\epsilon}=0.
 \label{eq:app-weak-bianchi}
\end{equation}
The Einstein tensors converge uniformly and the connections and volume forms
converge in the corresponding lower norms.  Passing to the limit proves
$\nabla^aG_{ab}=0$ distributionally across the join.

\section{Fluid mixtures and a regular mass profile}
\label{app:profiles}

For a finite collection of noninteracting positive-density fluids in a flat
background,
\begin{equation}
 H^2-H_\infty^2=\frac{8\pi}{3}
 \sum_j\rho_{j0}a^{-3(1+w_j)},
 \qquad \rho_{j0}>0.
 \label{eq:multifluid-friedmann}
\end{equation}
Let $w_{\rm slow}$ be the smallest $w_j$ among the nonvanishing components.
Factoring out the slowest scale gives
\begin{equation}
 H^2-H_\infty^2=a^{-3(1+w_{\rm slow})}
 \left[\frac{8\pi}{3}
 \sum_{w_j=w_{\rm slow}}\rho_{j0}+o(1)\right].
 \label{eq:multifluid-leading}
\end{equation}
The bracket has a positive leading coefficient, including when several
components share the same exponent.  Since $a\sim e^{H_\infty t}$, the
leading decay rate is
\begin{equation}
 \lambda_{\rm slow}=3(1+w_{\rm slow})H_\infty.
 \label{eq:multifluid-lambda}
\end{equation}
Replacing $w$ by $w_{\rm slow}$ in Eq.~\eqref{eq:main-p} therefore gives the
regularity exponent of the mixture.

The role of the limiting black-hole profile can be displayed with the Hayward
mass function \cite{Hayward2006}
\begin{equation}
 \mu_H(x)=\frac{x^3}{x^3+2q_H},
 \qquad q_H=\left(\frac{L}{m}\right)^2,
 \label{eq:hayward-profile}
\end{equation}
and the limiting lapse
\begin{equation}
 f_H(x)=1-\frac{2\mu_H(x)}x-h_\infty^2x^2.
 \label{eq:hayward-f}
\end{equation}
If $x_H$ is the simple seed-connected black-hole root, its dimensionless
surface gravity is
\begin{align}
 \bar\kappa_H&=\frac12 f_H'(x_H)\nonumber\\
 &=\frac{\mu_H(x_H)}{x_H^2}
 -\frac{\mu_H'(x_H)}{x_H}-h_\infty^2x_H,
 \label{eq:hayward-kappa}
\end{align}
where
\begin{equation}
 \mu_H'(x)=\frac{6q_Hx^2}{(x^3+2q_H)^2}.
 \label{eq:hayward-mu-prime}
\end{equation}
The local $p=2$ threshold is then
\begin{equation}
 w_{2,H}=\frac{2\bar\kappa_H}{3h_\infty}-1.
 \label{eq:hayward-threshold}
\end{equation}
For $h_\infty^2=1/100$ and $q_H=1/8$, root isolation on the branch connected
to the Schwarzschild root gives $w_{2,H}\simeq0.305514$.  Radiation lies above
this local threshold and below the McVittie value in
Eq.~\eqref{eq:default-thresholds}.  The comparison in
Fig.~\ref{fig:thresholds} isolates the effect of the limiting root slope; its
global interpretation remains specific to the spacetime realizing the chosen
profile.  Exact cosmological embeddings with regular horizons provide a
broader physical context for such profiles \cite{CadoniEtAl2026}.

Finally, the exponent in the asymptotic analysis of
Ref.~\cite{KaloperKlebanMartin2010} maps to the present notation through
\begin{equation}
 \frac{H_\infty}{\alpha_{\rm KKM}}=\bar\kappa,
 \qquad \alpha_{\rm KKM}\delta=p,
 \label{eq:KKM-map}
\end{equation}
so the relevant second-derivative power is $p-2$.  The decay constant denoted
$B$ in the special-rate analysis of
Ref.~\cite{DaSilvaFontaniniGuariento2013} equals $\bar\kappa$; its resonant
case is $p=1$.

\section*{Data availability}
No observational or experimental data were used in this work.  Symbolic
checks and the script generating Fig.~\ref{fig:thresholds} accompany the
source files.

\bibliography{references}

\end{document}